\documentclass{article}

\usepackage{graphicx}
\usepackage{amsmath}
\usepackage{amsfonts}
\usepackage{amssymb}
\usepackage{amsthm}
\usepackage{booktabs}
\usepackage{float}
\usepackage{subcaption}
\usepackage[
font={small,it},     
labelfont=bf,        
labelsep=period,
justification=justified,
singlelinecheck=false  
]{caption}

\usepackage[round]{natbib}

\usepackage[margin=1in]{geometry}
\usepackage{authblk}
\usepackage{colortbl}
\usepackage[dvipsnames]{xcolor} 
\usepackage[
colorlinks=true,
linkcolor=blue,        
citecolor=blue!60!black, 
urlcolor=blue!60!black,  
filecolor=black
]{hyperref}

\usepackage{tikz}
\usetikzlibrary{shapes,positioning,fit,backgrounds,calc}

\usetikzlibrary{arrows.meta}

\newtheorem{assumption}{Assumption}[section]
\newtheorem{theorem}[assumption]{Theorem}

\theoremstyle{plain}  

\theoremstyle{definition}  

\theoremstyle{remark}  
\newtheorem{remark}[assumption]{Remark}

\newcommand{\E}{\mathbb{E}}  
\newcommand{\prob}{\mathbb{P}}  
\newcommand{\indicator}{\mathbb{I}}
\newcommand{\primetreatmentunits}[1]{\mathcal{T}^{\text{prim}}_{#1}}
\newcommand{\mixeoutcomeunits}{\mathcal{O}_{\text{Both}}}
\newcommand{\exposedcustomers}[1]{\mathcal{O}_{#1}}
\newcommand{\primaryoutcomeunits}{\mathcal{O}_{\text{prim}}}
\newcommand{\treatment}{\text{treatment}}
\newcommand{\outcome}{\text{outcome}}
\newcommand{\primary}{\text{prim}}
\newcommand{\secondary}{\text{sec}}
\newcommand{\avg}{\text{avg}}

\newcommand{\bZ}{\mathbf{Z}}
\newcommand{\bw}{\mathbf{w}}

\newcommand{\Edirect}{E^{\text{Dir}}} 
\newcommand{\Eindirect}{E^{\text{Ind}}}

\newcommand{\primaryset}{\mathcal{T}_{\text{prim}}}
\newcommand{\secondaryset}{\mathcal{T}_{\text{sec}}}
\newcommand{\PTTE}{\text{PTTE}}
\newcommand{\STTE}{\text{STTE}}

\definecolor{primarycolor}{RGB}{70, 130, 180}  
\definecolor{secondarycolor}{RGB}{128, 128, 128}  
\definecolor{outcomecolor}{RGB}{46, 125, 50}  
\definecolor{mixedcolor}{RGB}{255, 152, 0}  
\definecolor{subsetcolor}{RGB}{156, 39, 176}  
\definecolor{othercolor}{RGB}{0, 0, 0}  

\title{Estimating Total Effects in Bipartite Experiments with Spillovers and Partial Eligibility}

\author{Albert Tan{\textsuperscript{*}}\quad
	Mohsen Bayati{\textsuperscript{*,$\dagger$}}\quad
	James Nordlund{\textsuperscript{*}}\quad
	Roman Istomin{\textsuperscript{*}}}
\date{}

\begin{document}
	\maketitle
	\begingroup
	\renewcommand\thefootnote{}
	\footnotetext{\textsuperscript{*}\;Amazon \quad \textsuperscript{$\dagger$}\;Stanford University}
	\addtocounter{footnote}{-1}
	\endgroup

	\begin{abstract}
		We study randomized experiments in bipartite systems where only a subset of treatment-side units are eligible for assignment while all units continue to interact, generating interference. We formalize \emph{eligibility-constrained bipartite experiments} and define estimands aligned with full deployment: the \emph{Primary Total Treatment Effect} (PTTE) on eligible units and the \emph{Secondary Total Treatment Effect} (STTE) on ineligible units. Under randomization within the eligible set, we give identification conditions and develop interference-aware ensemble estimators that combine exposure mappings, generalized propensity scores, and flexible machine learning. We further introduce a projection that links treatment- and outcome-level estimands; this mapping is exact under a \emph{Linear Additive Edges} condition and enables estimation on the (typically much smaller) treatment side with deterministic aggregation to outcomes. In simulations with known ground truth across realistic exposure regimes, the proposed estimators recover PTTE and STTE with low bias and variance and reduce the bias that could arise when interference is ignored. Two field experiments illustrate practical relevance: our method corrects the direction of expected interference bias for a pre-specified metric in both studies and reverses the sign and significance of the primary decision metric in one case.  
	\end{abstract}

	\section{Introduction}\label{sec:intro}
	
	Traditional A/B testing assumes the Stable Unit Treatment Value Assumption (SUTVA), under which each unit’s outcome depends only on its own treatment \citep{cox1958planning}. Interference or spillover, the violation of this assumption, is ubiquitous in networked settings. Seminal contributions clarified identification and inference under interference \citep{halloran1995causal,sobel2006what,hudgens2008toward}; for a concise review of design and analysis choices when SUTVA fails, see \citep{athey2017econometrics}.
	
	In many applications there is a natural bipartite (two-sided) structure underlying the interference. For example, in ride-sharing services the \textit{treatment units} may be drivers (e.g., a new routing or dispatch policy) while outcomes are measured on riders (wait times, booking rates, satisfaction). A rider may interact with multiple drivers, and a driver’s treatment can affect not only their riders but also other nearby drivers and their riders, creating interference patterns that invalidate standard A/B analyses. Importantly, only a \emph{subset} of drivers may be eligible for treatment (e.g., those in specific cities, app versions, or fleet programs). Even when assignment is restricted to these eligible drivers, ineligible drivers remain in the system and continue to affect matching, congestion, and learning dynamics.
	
	We study this problem by extending the bipartite experimental design framework of \cite{doudchenko2020causal} to \emph{eligibility-constrained bipartite experiments}, where randomization is feasible only among a subset of treatment-side units but interference propagates through all units. Our design formalizes (i) eligibility sets and exposure mappings that allow ineligible units to affect outcomes, and (ii) identification under randomization within the eligible set. Methodologically, we incorporate rich interaction data between treatment and outcome units, build on generalized propensity score ideas of \cite{hirano2004propensity}, and leverage modern machine learning to flexibly model exposure–response relationships.
	
	\subsection{Intuition Behind Primary and Secondary Treatment Effects}
	
	Consider a hypothetical ride-sharing service that wants to test a new app design feature: enhanced visual prominence and preferential placement for economy vehicles in the rider interface. The service operates with multiple vehicle categories: economy cars (compact sedans), premium vehicles (luxury sedans), and XL vehicles (SUVs). The treatment involves displaying economy vehicles at the top of the selection screen with highlighted badges indicating ``Quick Pickup'' or ``Eco-Friendly Choice,'' enhanced driver profiles showing ratings and estimated arrival times more prominently, and subtle visual animations that draw attention to these options. Due to technical constraints or strategic considerations, this enhanced presentation can only be applied to economy vehicles, while premium and XL vehicles maintain their standard display format.
	
	In this setting, economy vehicles represent our primary units; they are eligible for the enhanced app treatment and can receive the preferential placement. Premium and XL vehicles constitute our secondary units; while they cannot receive the enhanced presentation directly, they remain active in the system and may experience indirect effects. When economy vehicles become more visually prominent and appear first in the selection interface, rider behavior shifts in many ways. Some riders who might have scrolled down to select premium vehicles now choose the prominently displayed economy option. Others, particularly those looking for luxury or larger vehicles, may actually become more likely to select premium or XL options as a conscious reaction against the service's apparent steering toward economy vehicles.
	
	We introduce the notion of Primary Total Treatment Effect (PTTE) that captures what happens to economy vehicles when we compare two scenarios: one where all economy vehicles receive the enhanced app treatment versus one where none do. This effect encompasses not only the direct impact on economy vehicle utilization rates but also the competitive dynamics that emerge. For instance, when all economy vehicles have enhanced placement, riders cannot simply scroll past to find a ``normal'' economy option; they must either select from the prominently featured economy vehicles or deliberately choose a different vehicle category. The PTTE thus measures the true impact on the primary units when the feature is fully deployed.
	
	Similarly, we introduce the notion of Secondary Total Treatment Effect (STTE) that quantifies the impact on premium and XL vehicles. Even though these vehicles never receive the enhanced app treatment, their outcomes may change substantially. Premium vehicles might see decreased demand as riders are steered toward the more prominent economy options. Conversely, they might experience increased utilization from riders who interpret the economy vehicle promotion as a signal of lower quality or longer wait times, leading them to ``upgrade'' their choice. XL vehicles might capture riders traveling in groups who previously would have booked multiple economy vehicles but now, seeing the service's emphasis on individual economy rides, opt for a single larger vehicle instead. The magnitude and direction of these effects depend on complex behavioral patterns, user interface design principles, and rider preferences that vary across time and geography.
	
	The main insight is that traditional experimentation, which might randomly assign the enhanced app treatment to half of economy vehicles, creates an artificial environment. Riders see an inconsistent interface where some economy vehicles appear with enhanced prominence while others appear normally, potentially causing confusion and unnatural selection patterns. This leads to misestimating the treatment effect; at full deployment, the consistent interface creates different choice architecture and behavioral responses. The methodology studied in this paper aims to correct for these biases by explicitly modeling the bipartite structure of rider-driver interactions and the eligibility constraints that define which units can receive the app treatment.
	
	\subsection{Related Work}
	
	The foundations of causal inference under interference draw on early epidemiologic formulations for vaccine trials \citep{halloran1995causal} and potential-outcomes frameworks that formalized estimands under partial interference and group randomization \citep{sobel2006what,hudgens2008toward}, as well as neighborhood-based interference \citep{sussman2017elements}. Subsequent work developed randomization-based tests and design principles \citep{rosenbaum2007interference}, generalized estimands and inverse-probability estimators for arbitrary interference \citep{tchetgen2012interference,aronow2017general}, and clarified what standard estimators target when interference is unknown \citep{savje2021unknown}. These threads motivate network-aware designs such as graph-cluster randomization and bias-reducing analyses \citep{ugander2013graph,eckles2017design}, which our bipartite and eligibility-constraint setting builds upon. 
	
	More recent work considers settings where the interference network is only partially observed or entirely unknown, developing designs and estimators that remain valid in such cases \citep{agarwal2023network,yu2022estimating,cortez2022staggered,shirani2024causal,shirani2025can}. Within the class of experiments where a bipartite interaction structure between treatment and outcome units is observed, a growing literature studies design and analysis of such experiments. \citet{pougetabadie2019variance} propose correlation-clustering designs that reduce variance under linear exposure mappings, and \citet{harshaw2023bipartite} introduce the Exposure Reweighted Linear estimator and associated clustering strategies for unbiased and asymptotically normal total effect estimation. \citet{brennan2022cluster} study cluster-randomized designs in one-sided bipartite settings, and \citet{shi2024scalable} develop scalable randomization-based inference and covariate-adjusted estimators for large-scale bipartite experiments. Our paper builds directly on \citet{doudchenko2020causal}, who introduced the bipartite framework and generalized propensity score ideas for two-sided experiments, and is complementary to these design- and analysis-focused contributions.
	
	Beyond randomized experiments, there is emerging work on bipartite interference in quasi-experimental and longitudinal settings. \citet{chen2024differenceindifferences} develop a difference-in-differences framework under bipartite network interference with staggered adoption, reconfiguring panel data so that analysis is carried out at the intervention-unit level while preserving interpretation at the outcome level. Conceptually, their mapping between levels is related to our projection between treatment- and outcome-level estimands, but our focus is on randomized experiments with partial eligibility and explicit primary and secondary total effects.

	A related line of work examines experimental design and equilibrium responses in large two-sided systems. \citet{bajari2023experimental} survey multiple randomization designs and interference-aware decision metrics, while \citet{holtz2020reducing} analyze interference bias in pricing experiments. \citet{johari2022experimental} and \citet{munro2021treatment} study equilibrium adjustments to interventions on one side of a service and the implications for experimentation. Our framework differs by imposing explicit eligibility constraints on treatment-side units and defining primary and secondary total treatment effects together with a projection mapping between treatment- and outcome-level estimands.

	There is also a broad methodological literature on interference and spillovers spanning statistics, econometrics, and related fields. Representative contributions include design-based inference and randomization tests under general interference \citep{athey2018exact,basse2019randomization,puelz2022graph}, designs for networked and cluster experiments \citep{jagadeesan2020designs,baird2018optimal}, two-stage settings and noncompliance \citep{basse2018analyzing,imai2021causal}, identification under partial or unknown network exposure \citep{leung2022causal,egami2021spillover,forastiere2021identification}, and estimation of direct and indirect effects \citep{hu2022average,choi2017monotone}. Additional perspectives include auto-$g$-computation on networks \citep{tchetgen2021auto}, peer encouragement designs \citep{kang2016peer}, and frameworks for bipartite interference \citep{zigler2021bipartite}. For broader syntheses of interference and network-aware experimentation, see \citet{ogburn2014diagrams,vanderweele2013social,benjaminchung2018spillover,keele2022spillover,athey2017econometrics}.
	
	We contribute along two dimensions. First, we formalize eligibility-constrained bipartite designs in which only a subset of treatment-side units can be randomized while ineligible units still participate in interference; we define estimands (PTTE and STTE) and give identification under randomization within the eligible set. Second, we develop and analyze a projection mapping between treatment- and outcome-level total effects that holds under a \emph{Linear Additive Edges} assumption. This mapping preserves the estimand, enables estimation on the typically much smaller set of treatment units with projection-based inference to the outcome level, and delivers substantial computational gains without sacrificing interpretability.

	\section{Methodology}\label{sec:methods}
	
	This section develops our framework for eligibility-constrained bipartite experiments. For ease of reference, all notation is summarized in Table~\ref{tab:notation} in the Appendix. We begin by formalizing the experimental setup and introducing the distinction between primary (eligible) and secondary (ineligible) treatment units (Section~\ref{subsec:setup}). We then develop ensemble estimators at two levels: first for outcome units like riders (Section~\ref{subsec:outcome-ensemble}) and then for treatment units like drivers (Section~\ref{subsec:treatment-unit-ensemble}), defining the Primary Total Treatment Effect (PTTE) at each level. Section~\ref{subsec:assumptions} presents the formal assumptions required for identification. Our projection method (Section~\ref{subsec:projection}) links treatment effects across levels under a linear additive edges assumption and enables computationally efficient estimation. Section~\ref{subsec:estimation} details our machine learning-based estimation procedures that leverage generalized propensity scores and network structure. Finally, Section~\ref{subsec:secondary} extends the framework to quantify Secondary Total Treatment Effects (STTE) on ineligible units, capturing the full ecosystem impact of the intervention.
	
	\subsection{Problem setup}\label{subsec:setup}
	
	Consider a bipartite experiment with two types of units: treatment units (e.g., drivers in a ride-sharing service) and outcome units (e.g., riders), where outcomes are observed for both types. In our ride-sharing example, enhanced app placement may be deployed to certain vehicle types, affecting both driver utilization and rider selection patterns.
	We distinguish between two categories of treatment units, as illustrated in Figure \ref{fig:bipartite-graph}. Primary units ($\primaryset$) are those eligible for treatment. Secondary units ($\secondaryset$), while not eligible for direct treatment, may be affected through spillover effects due to shared connections with outcome units. The right side of Figure \ref{fig:bipartite-graph} represents the outcome units, with connections between treatment and outcome units shown as edges in the graph. We denote $\primetreatmentunits{i}$ as the set of primary units connected to outcome unit $i$, and conversely, $\exposedcustomers{j}$ represents the set of outcome units connected to treatment unit $j$, and $\primaryoutcomeunits$ is the set of outcomes connected to at least one primary treatment unit.
	The overlap region $\mixeoutcomeunits$ captures outcome units exposed to both primary and secondary treatment units. 
	The matrix $\bw$ captures the relationship structure, where $w_{ij}$ represents the connection strength between outcome unit $i$ and treatment unit $j$.
	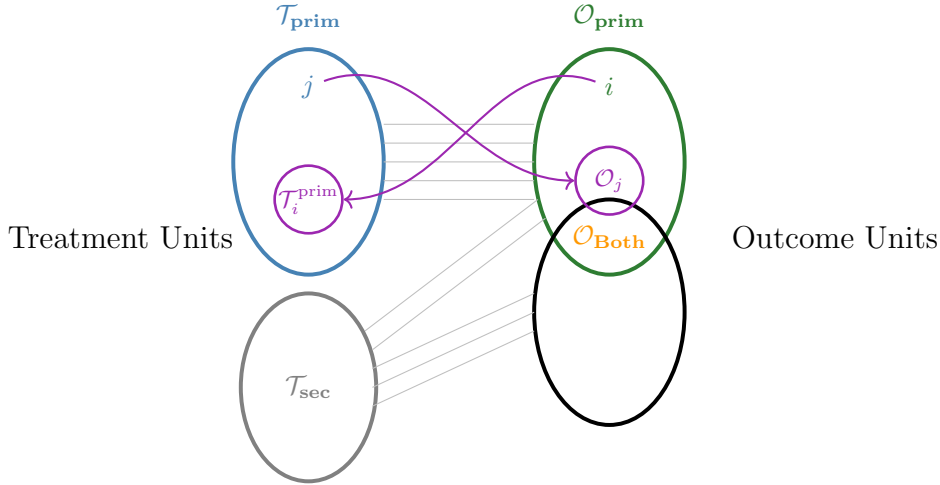
\begin{figure}
		\centering
		
		\begin{tikzpicture}[scale=0.5]
			\tikzstyle{primaryset} = [draw=primarycolor, ellipse, minimum width=2cm, minimum height=3cm, line width=1.5pt]
			\tikzstyle{secondaryset} = [draw=secondarycolor, ellipse, minimum width=1.8cm, minimum height=2.5cm, line width=1.5pt]
			\tikzstyle{outcomeset} = [draw=outcomecolor, ellipse, minimum width=2cm, minimum height=3cm, line width=1.5pt]
			\tikzstyle{mixedset} = [draw=othercolor, ellipse, minimum width=2cm, minimum height=3cm, line width=1.5pt]
			\tikzstyle{subset} = [draw=subsetcolor, circle, minimum size=.9cm, line width=1pt]
			
			\node[primaryset] (Aprim) at (0,0) {};
			\node[above=0.1cm of Aprim, text=primarycolor, font=\bfseries] {$\primaryset$};
			
			\node[secondaryset] (As) at (0,-6) {};
			\node[text=secondarycolor, font=\bfseries] at ($(As)$) {$\secondaryset$};
			
			\node[outcomeset] (Cprim) at (8,0) {};
			\node[above=0.1cm of Cprim, text=outcomecolor, font=\bfseries] {$\primaryoutcomeunits$};
			
			\node[mixedset] (Cboth) at (8,-4) {};
			\node[text=mixedcolor, font=\bfseries] at ($(Cboth)+(0,2)$) {$\mixeoutcomeunits$};
			
			\node[subset] (Ac) at (0,-1) {};
			\node[text=subsetcolor, font=\small] at (Ac) {$\primetreatmentunits{i}$};
			
			\node[subset] (Ca) at (8,-.5) {};
			\node[text=subsetcolor, font=\small] at (Ca) {$\exposedcustomers{j}$};
			
			\node[text=primarycolor, font=\bfseries] (a) at (0,2) {$j$};
			\node[text=outcomecolor, font=\bfseries] (c) at (8,2) {$i$};
			
			\begin{scope}[gray!50, thin]
				\foreach \y in {-1,-0.5,0,0.5,1} {
					\draw (2,\y) -- (6,\y);
				}
				
				\draw (1.5,-4.5) -- (6.1,-1);
				\draw (1.7,-5) -- (6.3,-1.5);
				
				\foreach \y in {-5.5,-6,-6.5} {
					\draw (1.7,\y) -- (6,\y+2);
				}
			\end{scope}
			
			\draw[subsetcolor, thick, ->] (a) to[out=20, in=180] (Ca);
			\draw[subsetcolor, thick, ->] (c) to[out=160, in=0] (Ac);
			
			\node[font=\large] at (-5,-2) {Treatment Units};
			\node[font=\large] at (14,-2) {Outcome Units}; 
			
		\end{tikzpicture}
		\caption{Bipartite structure in eligibility-constrained experiments. Primary units $\primaryset$ (top left) are eligible for treatment; secondary units $\secondaryset$ (bottom left) are ineligible but connected to outcomes. Outcome units $\primaryoutcomeunits$ (top right) connect to primaries; $\mixeoutcomeunits$ (overlap) connect to both.}\label{fig:bipartite-graph}
	\end{figure}

	\paragraph{Edge-level representation and additive metrics.}
	In many applications of interest, including our motivating ride-sharing example, 
	it is natural to think of outcomes as aggregating over interactions along edges 
	of the bipartite graph: each ride corresponds to a rider–driver edge, each 
	impression to an advertiser–viewer edge, and so on. We formalize this by 
	introducing edge-level potential outcomes $Y_{ij}(\mathbf Z)$ for every 
	existing edge $(i,j)$, with $Y_{ij}(\mathbf Z)=0$ when $(i,j)$ is absent. 
	We then define
	\[
	Y_i(\mathbf Z) = \sum_j Y_{ij}(\mathbf Z), 
	\qquad
	Y_j(\mathbf Z) = \sum_i Y_{ij}(\mathbf Z),
	\]
	so that each unit's outcome is the sum over its incident edges.
	Many business metrics (counts of rides, orders, impressions, revenue, and 
	their additively weighted variants) are naturally of this form.
	
	\paragraph{Primary and secondary outcome components.}
	Let $\mathcal T_{\text{prim}}$ and $\mathcal T_{\text{sec}}$ denote the sets 
	of primary (eligible) and secondary (ineligible) treatment units, respectively. 
	For any outcome unit $i$ we 
	decompose its outcome as
	\[
	Y_i(\mathbf Z) 
	= \sum_{j \in \mathcal T_{\text{prim}}} Y_{ij}(\mathbf Z)
	+ \sum_{j \in \mathcal T_{\text{sec}}} Y_{ij}(\mathbf Z)
	=: Y_{i,\primary}(\mathbf Z) + Y_{i,\secondary}(\mathbf Z).
	\]
	Here $Y_{i,\primary}(\mathbf Z)$ collects the contribution coming through 
	primary treatment units, while $Y_{i,\secondary}(\mathbf Z)$ collects the 
	contribution coming through secondary units. In the ride-sharing example, 
	$Y_{i,\primary}$ would be the part of rider $i$'s metric attributable to trips 
	taken with primary (eligible) vehicles, and $Y_{i,\secondary}$ the part 
	attributable to secondary vehicle categories.

	\subsection{Outcome-Unit Ensemble}\label{subsec:outcome-ensemble} 
	
	For each outcome unit $i$, we define exposure as: 
	\begin{equation}\label{eq:exposure}
		E_i(\bZ) = \sum_{j \in \primetreatmentunits{i}} w_{ij} Z_j\,,
	\end{equation}
	where $\primetreatmentunits{i}$ is the set of primary units connected to $i$, $Z_j$ indicates treatment assignment of unit $j$, and $\bZ$ is a vector of size $|\primaryset|+|\secondaryset|$, containing treatment assignment of all units. Note that treatment assignment of secondary units is always equal zero because they are not eligible for treatment, i.e., the last $|\secondaryset|$ coordinates of $\bZ$ are always equal to zero.
	The outcome specification becomes:
	\begin{equation}
		Y_{i,\primary}(\bw,E_i,X_i) = \Phi\Big(n^{\primary}_i, E_i, r(E_i,\bw,X_i)\Big) + \epsilon_i\label{eq:outcome-unit-specification}
	\end{equation}
	where $n^{\primary}_i$ is the number of primary units connected to $i$, $r(\cdot)$ is the generalized propensity score, $X_i$ are covariates, $\epsilon_i$ is the error term, and
	$\Phi$ is the outcome function. In the implementations, we utilize binomial expression for the propensity score which effectively specializes to the case $w_{ij}=1/n^{\primary}_i$ which means $E_i(\bZ)$ would be the fraction of treated primary neighbors.

	Let $\mathbf Z^{(1)}$ denote the assignment 
	where all primary units are treated (and secondary units are ineligible), and 
	$\mathbf Z^{(0)} \equiv \mathbf 0$ denote the assignment where no units are treated. It is straightforward to see that under $\mathbf Z^{(1)}$ all exposures in $\primaryoutcomeunits$ are equal to $1$ and under $\mathbf Z^{(0)}$ all such exposures are equal to $0$. Hence, at the outcome level, we define the Primary Total Treatment Effect (PTTE) as
	\begin{align}
		\PTTE_{\outcome} 
		&= \frac{1}{|\mathcal O_{\text{prim}}|}
		\sum_{i\in\mathcal O_{\text{prim}}}
		\Big(
		\E\big[ Y_{i,\primary}(\mathbf Z^{(1)}) \big]
		- 
		\E\big[ Y_{i,\primary}(\mathbf Z^{(0)}) \big]
		\Big)\nonumber\\
		&=\frac{1}{|\primaryoutcomeunits|}\sum_{i\in\primaryoutcomeunits}\left(\E\Big[Y_{i,\primary}(E_i=1)\Big] - \E\Big[Y_{i,\primary}(E_i=0)\Big]\right)\,. \label{eq:PTTE-outcome-primary}
	\end{align}
	That is, $\PTTE_{\outcome}$ captures the change in the primary-edge component 
	of outcome-unit metrics when we move from no primary units treated to all 
	primary units treated. 
	Note that, expectation (with notation $\E$) is taken with respect to all randomness. We use $Y_{i,\primary}(e)$ to denote the potential outcome of unit $i$ under the exposure level $e$, and model $Y_{i,\primary}(e)$ via \eqref{eq:outcome-unit-specification} with $\bw$ and $X_i$ being fixed throughout. This specification corresponds to an \textit{exposure mapping assumption}: for any two assignments $\bZ$ and $\bZ'$ with $E_i(\bZ)=E_i(\bZ')$, we have $Y_{i,\primary}(\bZ)=Y_{i,\primary}(\bZ')$
	
	\subsection{Treatment-Unit Ensemble}\label{subsec:treatment-unit-ensemble}
	
	Denoting indicator function of an event $A$ by $\indicator(A)$, at the treatment-unit level, we define direct exposure by 
	\begin{equation}\label{eq:exposure-dir}
		\Edirect_j(\bZ) = Z_j\sum_{i\in\primaryoutcomeunits} \indicator(j\in \primetreatmentunits{i})\,,
	\end{equation}
	and indirect exposure by
	\begin{equation}\label{eq:exposure-ind}
		\Eindirect_j(\bZ)= \sum_{i\in\primaryoutcomeunits}\Big[\indicator(j\in\primetreatmentunits{i}) \sum_{k\neq j} Z_k \cdot \indicator(k\in \primetreatmentunits{i})\Big]\,.
	\end{equation}

	At the treatment-unit level, utilizing
	$
	Y_j(\mathbf Z) = \sum_i Y_{ij}(\mathbf Z)
	$, we define
	\begin{equation}\label{eq:PTTE-treatment-v0}
		\PTTE_{\treatment}
		= \frac{1}{|\mathcal T_{\text{prim}}|}
		\sum_{j\in\mathcal T_{\text{prim}}}
		\Big(
		\E\big[ Y_j(\mathbf Z^{(1)}) \big]
		- 
		\E\big[ Y_j(\mathbf Z^{(0)}) \big]
		\Big),
	\end{equation}
	i.e., the average total effect of full deployment on primary treatment units. Under our paramteric representation $Y_j(\Edirect_j(\bZ),\Eindirect_j(\bZ),X_j)$,
	\eqref{eq:PTTE-treatment-v0} can be written as,
	\begin{equation}\label{eq:PTTE-treatment}
		\PTTE_{\treatment} = \frac{1}{|\primaryset|}\sum_{j\in\primaryset}\E\left[\Psi\Big(\Edirect_j(\mathbf Z^{(1)}),\Eindirect_j(\mathbf Z^{(1)}), X_j\Big) - \Psi\Big(\Edirect_j(\mathbf Z^{(0)}),\Eindirect_j(\mathbf Z^{(0)}), X_j\Big)\right]\,,
	\end{equation}
	where it is easy to see that $\Edirect_j(\mathbf Z^{(0)})$ and $\Eindirect_j(\mathbf Z^{(0)})$ are both equal to $0$.

	\subsection{Formal Assumptions}\label{subsec:assumptions}
	
	Our framework relies on three key assumptions adapted from \cite{doudchenko2020causal}.
	\begin{assumption}[Exogenous Network]\label{ass:fixed-weights}
		$\bw$ is not affected by the treatment assignment $\bZ$.
	\end{assumption}
	This assumption requires that the connections between outcome and treatment units remain stable throughout the experiment. In our application, this means treatment does not alter which outcome units interact with which treated units.
	
	\begin{assumption}[Weak Unconfoundedness]\label{ass:weak-unconfoundedness}
		For all exposure levels $e$:
		$\indicator[E_i = e] \perp Y_i(e) | \bw_i$
		where $\bw_i$ represents outcome unit $i$'s connection weights to all treatment-side units.
	\end{assumption}
	
	Under randomization of treated units, this assumption holds by design, as exposure is determined by the random treatment assignment and the fixed network structure.
	
	\begin{assumption}[Overlap]\label{ass:overlap}
		For all exposure levels $e$ and weight vectors $\bw$:
		$0 < \prob(E_i = e|\bw) < 1\,.$
	\end{assumption}
	
	This positivity condition ensures sufficient variation in exposure levels for identification. In practice, it requires that the treatment probability and network structure create adequate overlap in the exposure distribution.
	
	The final assumption, already discussed at the beginning of this section, formalizes additivity of outcomes along the edges of the bipartite graph.
	
	\begin{assumption}[Linear additive edges]\label{ass:linear-additive}
		For all assignments $\mathbf Z$, unit-level outcomes can be written as
		\[
		Y_i(\mathbf Z) = \sum_j Y_{ij}(\mathbf Z),
		\qquad
		Y_j(\mathbf Z) = \sum_i Y_{ij}(\mathbf Z),
		\]
		where $Y_{ij}(\mathbf Z)$ is the potential outcome associated with edge 
		$(i,j)$ and $Y_{ij}(\mathbf Z)=0$ if $(i,j)$ is absent.
	\end{assumption}
	
	Assumption~\ref{ass:linear-additive} holds exactly for additive metrics such as 
	total rides, completed orders, or revenue, and for any linear reweighting of 
	these (e.g., weighted revenue). It is not appropriate for inherently 
	non-additive metrics (medians, quantiles, capped scores).

	\subsection{Projection Between Treatment and Outcome Effects}\label{subsec:projection}
	
	We formalize the link between the treatment- and outcome-level estimands when edge outcomes add linearly.
	
	Recall that
	$\PTTE_{\outcome}$ is defined in \eqref{eq:PTTE-outcome-primary} via the exposure $E_i$, which aggregates \emph{only} primary-unit assignments, and $\PTTE_{\treatment}$ is given in \eqref{eq:PTTE-treatment}.
	
	\begin{theorem}[Exact projection of PTTE under edge additivity]\label{thm:projection}
		Suppose Assumption~\ref{ass:linear-additive} holds. Then
		\begin{equation}\label{eq:proj_main}
			\PTTE_{\outcome} 
			= \frac{|\mathcal T_{\text{prim}}|}{|\mathcal O_{\text{prim}}|}
			\;\PTTE_{\treatment},
		\end{equation}
		where $\PTTE_{\outcome}$ is defined in \eqref{eq:PTTE-outcome-primary} in terms 
		of $Y_{i,\primary}(\mathbf Z)$ and $\PTTE_{\treatment}$ is defined in 
		\eqref{eq:PTTE-treatment} in terms of $Y_j(\mathbf Z)$.
	\end{theorem}
	\begin{proof}
		By Assumption~\ref{ass:linear-additive}, for any assignment $\bZ$ and any
		$i \in \primaryoutcomeunits$,
		\[
		Y_{i,\primary}(\bZ)
		= \sum_{j\in\primaryset} Y_{ij}(\bZ),
		\qquad
		Y_j(\bZ)
		= \sum_i Y_{ij}(\bZ).
		\]
		Using the definition in \eqref{eq:PTTE-outcome-primary} and linearity of
		expectation,
		\begin{align*}
			\PTTE_{\outcome}
			&= \frac{1}{|\primaryoutcomeunits|}
			\sum_{i\in\primaryoutcomeunits}
			\E\big[ Y_{i,\primary}(\bZ^{(1)}) - Y_{i,\primary}(\bZ^{(0)}) \big] \\
			&= \frac{1}{|\primaryoutcomeunits|}
			\sum_{i\in\primaryoutcomeunits}
			\E\Big[ \sum_{j\in\primaryset}
			\big( Y_{ij}(\bZ^{(1)}) - Y_{ij}(\bZ^{(0)}) \big) \Big] \\
			&= \frac{1}{|\primaryoutcomeunits|}
			\sum_{j\in\primaryset}
			\E\Big[ \sum_{i\in\primaryoutcomeunits}
			\big( Y_{ij}(\bZ^{(1)}) - Y_{ij}(\bZ^{(0)}) \big) \Big],
		\end{align*}
		where we interchange the (finite) sums over $i$ and $j$.
		
		Every edge $(i,j)$ with $j\in\primaryset$ necessarily has
		$i\in\primaryoutcomeunits$, so
		\[
		\sum_{i\in\primaryoutcomeunits} Y_{ij}(\bZ)
		= \sum_i Y_{ij}(\bZ)
		= Y_j(\bZ).
		\]
		Thus the inner sum equals $Y_j(\bZ^{(1)}) - Y_j(\bZ^{(0)})$, and
		\[
		\PTTE_{\outcome}
		= \frac{1}{|\primaryoutcomeunits|}
		\sum_{j\in\primaryset}
		\E\big[ Y_j(\bZ^{(1)}) - Y_j(\bZ^{(0)}) \big].
		\]
		By the definition of $\PTTE_{\treatment}$ in \eqref{eq:PTTE-treatment},
		\[
		\PTTE_{\treatment}
		= \frac{1}{|\primaryset|}
		\sum_{j\in\primaryset}
		\E\big[ Y_j(\bZ^{(1)}) - Y_j(\bZ^{(0)}) \big],
		\]
		so
		\[
		\PTTE_{\outcome}
		= \frac{|\primaryset|}{|\primaryoutcomeunits|}\;\PTTE_{\treatment},
		\]
		which is \eqref{eq:proj_main}.
	\end{proof}

	\begin{remark}
		In applications where one suspects displacement onto ineligible units, Section~\ref{subsec:secondary} provides STTE estimands; comparing the projected quantity in \eqref{eq:proj_main} to outcome-level estimates is an empirical check.
	\end{remark}
	
	\begin{remark}
		If a decision metric applies known edge weights $\omega_{ij}\!\ge\!0$, the same proof holds with $Y_{ij}$ replaced by $\omega_{ij}Y_{ij}$ and normalizations based on total $\omega$-mass.
	\end{remark}
	
	\begin{remark}
		For ratios or quantiles, a first-order linearization around the observed edge measure provides an approximate projection with a curvature remainder; see the discussion in Section~\ref{sec:conclusion}.
	\end{remark}
	
	\subsection{Estimation Methods}\label{subsec:estimation}
	
	Our estimation procedure utilizes the observed bipartite network structure to construct features (including exposure levels, number of exposed treatment-units, and propensity scores) which serve as inputs to machine learning models (linear, polynomial, non-parametric). These models estimate the outcome function $\Phi$ or $\Psi$ from which we derive corresponding PTTE estimates by comparing predicted outcomes at full versus zero exposure levels. 
	
	More formally, our estimation proceeds in three steps:
	
	\paragraph{Step 1: Feature Construction.} For each unit, construct:
	\begin{itemize}
		\item \emph{Exposure variables:} $E_i$, $\Edirect_j$, $\Eindirect_j$ as defined above
		\item \emph{Network features:} $n^{\primary}_i$ (number of connections), graph-based covariates
		\item \emph{Propensity scores:} When neighbors are unweighted ($w_{ij}=1/n^{\primary}_i$), $E_i$ takes values 
		in $\{0,1/n^{\primary}_i,\dots,1\}$. Writing $e = k/n^{\primary}_i$ with 
		$k\in\{0,\dots,n^{\primary}_i\}$, the generalized propensity score has the 
		binomial form
		\begin{equation}\label{eq:propensity}
			r(e, n^{\primary}_i, p)
			= \binom{n^{\primary}_i}{k} p^{k}(1-p)^{n^{\primary}_i-k}\,.
		\end{equation}
		For treatment-side units, we do not explicitly calculate the propensity score. Instead, we assume that the propensity score is captured through the non-parametric function of $\psi$.
	\end{itemize}
	
	\paragraph{Step 2: Outcome Function Estimation.} Using the constructed features, estimate $\Phi$ or $\Psi$ via machine learning:
	\begin{itemize}
		\item \emph{Linear polynomial (LP):} Second-order polynomial in $(E_i, r, E_i \times r, E_i^2, r^2)$ for outcome-level, similar for treated-unit level (we only use LP in synthetic settings where there are no covariates)
		\item \emph{Kernel Ridge Regression (KRR):} Non-parametric regression with Gaussian kernel, bandwidth selected via 5-fold cross-validation
		\item \emph{Other methods:} Random forests, gradient boosting (details in online supplement)
	\end{itemize}
	
	\paragraph{Step 3: Counterfactual Prediction and Aggregation.} For each unit, predict outcomes under full treatment ($E=1$ or $\Edirect_j(\mathbf{Z}^{(1)}), \Eindirect_j(\mathbf{Z}^{(1)})$ and no treatment ($E=0$ or $\Edirect_j(\mathbf{Z}^{(0)}), \Eindirect_j(\mathbf{Z}^{(0)})$). Average the differences to obtain PTTE estimates via Eq.~\eqref{eq:PTTE-outcome-primary} or \eqref{eq:PTTE-treatment}.
	
	\paragraph{Variance Estimation.} We construct confidence intervals via bootstrap: sample treatment-side units with replacement, re-estimate on each bootstrap sample, and compute quantiles of the resulting PTTE distribution. For the projection approach, we project each bootstrap replicate and compute quantiles at the target level.

	\subsection{Extension to Secondary Effects}\label{subsec:secondary}
	
	While primary units receive direct treatment, secondary units ($\secondaryset$) experience spillover effects through shared connections with outcome units. These secondary total treatment effects (STTE) capture the indirect impact on ineligible units, an important component for understanding the full ecosystem effects of eligibility-constrained experiments.
	
	\subsubsection{Outcome-Unit Level STTE}
	
	At the outcome-unit level, we extend our framework to account for connections to both primary and secondary units. Recall that $\mixeoutcomeunits$ denotes outcome units connected to both 
	primary and secondary treatment units. For $i\in\mixeoutcomeunits$, the 
	secondary-edge component $Y_{i,\secondary}(\mathbf Z)$ can change when we 
	treat the primary units. We define the Secondary Total Treatment Effect 
	(STTE) at the outcome level as
	\begin{equation}\label{eq:STTE-outcome-new}
		\STTE_{\outcome}
		= \frac{1}{|\mixeoutcomeunits|}
		\sum_{i\in \mixeoutcomeunits}
		\Big(
		\E\big[ Y_{i,\secondary}(\mathbf Z^{(1)}) \big]
		- 
		\E\big[ Y_{i,\secondary}(\mathbf Z^{(0)}) \big]
		\Big).
	\end{equation}
	Outcome units connected only to secondary treatment units are unaffected by 
	the randomization among primary units and therefore do not contribute to 
	$\STTE_{\outcome}$.
	
	In estimation we model $Y_{i,\secondary}$ via a function $\Phi_{\secondary}$ as follows. For outcome units in $\mixeoutcomeunits$, let $n^{\secondary}_i$ be the number of secondary treatment units connected to outcome unit $i$. We then assume:
	\begin{equation}\label{eq:sec-specification}
		Y_{i,\secondary}(\tilde{\bw},\bZ,X_i) = \Phi_{\secondary}\Big(n^{\primary}_i, n^{\secondary}_i, E_i, r(E_i,\tilde{\bw},X_i)\Big) + \epsilon_i
	\end{equation}
	where $\tilde{\bw}$ denotes the enriched graph structure including treatment unit classifications, i.e., primary and secondary. 
	Under the same exposure mapping assumption as in Section~\ref{subsec:outcome-ensemble}, 
	we can equivalently write
	\eqref{eq:STTE-outcome-new} as
	\begin{equation}\label{eq:STTE-outcome}
		\STTE_{\outcome} = \frac{1}{|\mixeoutcomeunits|}\sum_{i\in\mixeoutcomeunits}\E\Big[Y_{i,\secondary}(e_{i,\max})\Big] - \E\Big[Y_{i,\secondary}(0)\Big]\,,
	\end{equation}
	where $e_{i,\max}$ denotes the exposure of unit $i$ 
	when all primary units are treated.
	
	\subsubsection{Treatment-Unit Level STTE}
	
	For a secondary treatment unit $j\in\mathcal T_{\text{sec}}$, we focus on 
	the part of its outcome that flows through outcome units exposed to primary 
	units. Define
	\[
	Y_{j,\secondary}(\mathbf Z)
	:= \sum_{i\in\mathcal \mixeoutcomeunits} Y_{ij}(\mathbf Z).
	\]
	Edges connecting $j$ to outcome units that are never connected to any primary 
	unit cannot be affected by the primary treatment assignment and are thus 
	excluded from the STTE. We define the treatment-level STTE as
	\begin{equation}\label{eq:STTE-treatment-new}
		\STTE_{\treatment}
		= \frac{1}{|\mathcal T_{\text{sec}}|}
		\sum_{j\in\mathcal T_{\text{sec}}}
		\Big(
		\E\big[ Y_{j,\secondary}(\mathbf Z^{(1)}) \big]
		- 
		\E\big[ Y_{j,\secondary}(\mathbf Z^{(0)}) \big]
		\Big).
	\end{equation}
	In our implementation, $Y_{j,\secondary}(\mathbf Z)$ is modeled as a function 
	exposure and covariates. Specifically, under the representation,
	\begin{equation}
		Y_{j,\secondary} = \Psi_{\secondary}(\Eindirect_j, X_j) + \epsilon_j\,,
	\end{equation}
	where $\Eindirect_j$ follows a variant of Eq.~\eqref{eq:exposure-ind} for secondary units,
	Eq.~\eqref{eq:STTE-treatment} can be written as:
	\begin{equation}\label{eq:STTE-treatment}
		\text{STTE}_{\treatment} = \frac{1}{|\secondaryset|}\sum_{j\in\secondaryset}\E\left[\Psi_{\secondary}\Big(\Eindirect_j(\mathbf Z^{(1)}), X_j\Big) - \Psi_{\secondary}(\Eindirect_j(\mathbf Z^{(0)}), X_j)\right]\,.
	\end{equation}
	
	\subsubsection{Projection and Estimation}
	
	Under Assumption~\ref{ass:linear-additive}, the same counting argument as in 
	Theorem~\ref{thm:projection} yields
	\begin{equation}\label{eq:projection-STTE-new}
		\STTE_{\outcome}
		= \frac{|\mathcal T_{\text{sec}}|}{|\mixeoutcomeunits|}
		\; \STTE_{\treatment}.
	\end{equation}
	This enables efficient computation at the treatment-unit level with projection to outcome-level effects, maintaining the computational advantages discussed for PTTE.
	
	The ability to quantify STTE is important for bipartite experiments where ineligible units constitute a substantial portion of the ecosystem. In our ride-sharing example, understanding how visual prominence of economy vehicles affects premium and XL vehicles informs both immediate design decisions and long-term fleet composition strategies.
	
	\section{Validation on Simulated and Real Experimental Data}\label{sec:empirics}
	
	We validate our methodology through two complementary approaches. First, we conduct extensive simulations where ground truth effects are known, allowing us to assess both primary (PTTE) and secondary (STTE) treatment effects while evaluating the accuracy of our estimators and projection methods. Second, we apply our framework to two real experiments, focusing on PTTE estimation where business logic provides directional expectations for bias validation.
	
	\subsection{Simulation Results}\label{subsec:simulated}
	
	\subsubsection{Simulation Design}
	
	We simulate a bipartite system where outcome units (riders) interact with both primary and secondary treatment units (different vehicle types) through an app interface. The outcome function for rider $i$ selecting vehicle type $j$ follows:
	\begin{equation}
		f_{ij} = \alpha_j + \gamma_i \sum_{k \in \mathcal{E}_i} \beta_{jk} \log(v_k) + \epsilon_{ij}
	\end{equation}
	where $f_{ij}$ represents the selection outcome (e.g., booking probability or ride frequency), $\alpha_j$ is the baseline selection rate for vehicle type $j$, $\gamma_i$ captures rider-specific interface responsiveness, $\beta_{jk}$ represents the \emph{attention spillover coefficient} between vehicle types $j$ and $k$, $v_k$ is the \emph{visibility score} for vehicle type $k$ (with baseline $v_k = 1.0$ for standard display), and $\mathcal{E}_i$ is the set of vehicle types (primary and secondary) available to rider $i$.
	
	The attention spillover coefficients $\beta_{jk}$ encode how visual prominence affects selection patterns through cognitive attention mechanisms. The key patterns are:
	\begin{itemize}
		\item $\beta_{jj} < 0$: \emph{Prominence saturation effect}, excessive visual emphasis on a vehicle type can trigger adverse user reactions, where riders perceive over-promotion as pushy or interpret it as a negative quality signal
		\item $\beta_{jk} > 0$ for substitutes: \emph{Complementary positioning}, when economy vehicles are prominently featured, some premium vehicles benefit from appearing as the deliberate ``upgrade'' choice in contrast
		\item $\beta_{jk} < 0$ for competitors: \emph{Competitive attention}, direct competition for limited user attention, where prominence of one vehicle type draws cognitive resources away from others
	\end{itemize}
	
	The logarithmic transformation captures \emph{diminishing attention returns}: initial gains in visual prominence have larger effects than additional prominence, reflecting cognitive limits in processing visual stimuli and the natural saturation of attention capture.
	
	Treatment consists of enhanced app placement for primary units (economy vehicles), implemented as a 10\% increase in visibility score ($v_k = 1.1$). This could represent various interface changes: moving from position 3 to position 1 in the selection list, increasing badge size and animation intensity, or expanding the screen real estate allocated to these vehicles. Only primary units are eligible for this enhanced display treatment due to technical constraints or strategic considerations, while secondary units (premium and XL vehicles) maintain standard display formatting.
	
	We vary the average number of primary units per rider (2.8 to 8.0, representing market density) and treatment probability (40-50\%, representing rollout coverage), running 50 replications per configuration to ensure robust estimates. For each replication we compute ground-truth $PTTEs$ and $STTEs$ by evaluating the data generating process at $\mathbf{Z}^{(1)}$ and $\mathbf{Z}^{(0)}$.
	
	\subsubsection{Primary Treatment Effects (PTTE)}
	
	Table \ref{tab:ptte-outcome} presents PTTE estimates across five specifications. The Basic approach (a variant of difference-in-means) that ignores interferences consistently underestimates effects by 15-20\%, as it fails to account for positive interferences that increase outcomes. Non-parametric methods (KRR) achieve estimates within 0.5\% of ground truth across all specifications, while linear polynomial (LP) approaches degrade substantially as exposure complexity increases.
	
	\begin{table}[h]
		\caption{Outcome-Level PTTE Simulation Results: Median values across 50 replications}
		\label{tab:ptte-outcome}
		\centering
		\small
		\begin{tabular}{ccccccccc}
			\toprule
			\multicolumn{3}{c}{Specification} &  \multicolumn{4}{c}{Outcome Level} &  \multicolumn{2}{c}{Treatment to Outcome Level}  \\
			\cmidrule(r){1-3} \cmidrule(r){4-7} \cmidrule(r){8-9}
			No    & \(\avg(|\primetreatmentunits{i}|)\)    & \% Primary Treated & GT    & Basic & LP & KRR   & Proj. GT & Proj. KRR \\
			\midrule
			1     & 2.8         & 50\%                      & 0.5 & 0.42 & 0.51 & 0.50 & 0.5 & 0.5 \\    
			2     & 5.4         & 50\%                      & 0.97 & 0.81 & 0.92 & 0.97 & 0.97 & 0.97  \\
			3     & 8.0         & 50\%                      & 1.46 & 1.21 & 0.98 & 1.46 & 1.46 & 1.46       \\
			4     & 8.0         & 45\%                      & 1.46 & 1.19 & 0.79 & 1.46 & 1.46 & 1.46        \\
			5     & 8.0         & 40\%                      & 1.46 & 1.16 & 0.58 & 1.46 & 1.46 & 1.46       \\  
			\bottomrule
		\end{tabular}
	\end{table}
	
	Figure \ref{fig:boxplot} illustrates the distribution of estimates for Specification 1, top row of Table \ref{tab:ptte-outcome}, demonstrating that treatment-level estimates projected to outcome level (Proj. KRR) match direct estimation while offering computational advantages, critical when outcome units outnumber treatment units by orders of magnitude. Similar figures for all specifications are provided in Appendix \ref{sec:additiona-sims}
	
	\begin{figure}[h]
		\centering
		\includegraphics[width=\linewidth]{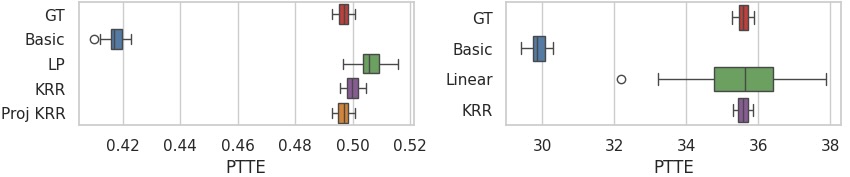}
		\caption{PTTE estimates for outcome units (left) and treatment units (right), showing ground truth (GT), Basic, linear polynomial (LP), Kernel Ridge Regression (KRR), and its projected variant (Proj. KRR, outcome level only).}
		\label{fig:boxplot}
	\end{figure}
	
	\subsubsection{Secondary Treatment Effects (STTE)}
	
	Table \ref{tab:stte-simulation} presents STTE estimates for secondary units. These effects, while smaller than PTTE, remain substantial, approximately 17\% of primary effects in our simulations. Importantly, STTE would be exactly zero under the no-spillover assumption, highlighting the magnitude of bias when ignoring network effects.
	
	\begin{table}
		\caption{STTE Simulation Results: Median values across 50 replications}
		\label{tab:stte-simulation}
		\centering
		\small  
		\begin{tabular}{cccccccc}
			\toprule
			\multicolumn{3}{c}{Specification} &  \multicolumn{3}{c}{Outcome-Level} &  \multicolumn{2}{c}{Treatment-Level}  \\
			\cmidrule(r){1-3} \cmidrule(r){4-6} \cmidrule(r){7-8}
			No    & \(\avg(|\primetreatmentunits{i}|)\)    & \% Primary Treated & GT     & XGBoost & Proj KRR & GT    & KRR \\
			\midrule
			1 & 2.8 & 50\% & 0.10 & 0.10 & 0.10 & 11.4 & 11.4 \\
			2 & 5.4 & 50\% & 0.17 & 0.17 & 0.17 & 24.0 & 24.0 \\
			3 & 8.0 & 50\% & 0.25 & 0.24 & 0.25 & 36.6 & 36.6 \\
			4 & 8.0 & 45\% & 0.25 & 0.24 & 0.25 & 36.6 & 36.6 \\
			5 & 8.0 & 40\% & 0.25 & 0.24 & 0.25 & 36.6 & 36.7 \\   
			\bottomrule
		\end{tabular}
	\end{table}

	For secondary effects, gradient boosting methods (XGBoost) perform well at the outcome level, while KRR excels at the treatment level. The projection approach again demonstrates remarkable accuracy, with projected estimates matching ground truth within 0.1\%. Similar distributional plots as in Figure \ref{fig:boxplot} are provided in Appendix \ref{sec:additiona-sims}. These results validate our framework's ability to capture the complete spillover ecosystem, quantifying both direct effects on eligible units and indirect effects on ineligible units that continue to operate in the bipartite setting.
	
	\subsubsection{Computational Benefits}
	
	The projection methodology offers substantial computational advantages. With typical ratios of outcome to treatment units exceeding 1000:1, estimation at the treatment level followed by projection reduces computation time from hours to minutes. For a system with 300,000 outcome units and 300 treatment units, we observe 1000x speedup while maintaining accuracy, needed for practical deployments.
	
	\subsection{Real Experiment Results}\label{subsec:real-data}
	Focusing on PTTE, we validate the method on two real experiments, denoted by A and B, with three pre-specified metrics, $M_1$, $M_2$, and $M_3$. 
	For both experiments, business logic and economic theory imply that ignoring spillovers should bias $M_1$ upward (direction known ex ante), whereas no directional prior is available for $M_2$ or $M_3$. Among the three, $M_3$ is the primary decision metric.
	
	\begin{table}[htp]
		\caption{Comparison of Basic Method (ignoring spillovers) and Our Approach on Two Real Experiments. For each anonymized metric, we report treatment effect estimates (ATE for Basic, PTTE for our method), statistical significance at 5\% level, and observed bias direction. Expected bias direction for $M_1$ follows from economic theory; $M_3$ is the primary decision metric. Sample sizes: $|\primaryset|\approx 7,000$ (Experiment A) and $8,000$ (Experiment B).}
		\label{tab:RealExperiments}
		\centering
		\small
		\begin{tabular}{@{}l@{\hskip 0.3cm}lcccccc@{}}
			\toprule
			& \textbf{Metric} & \multicolumn{2}{c}{\textbf{Basic Method}} & \multicolumn{2}{c}{\textbf{Our Approach}} & \multicolumn{2}{c}{\textbf{Bias}} \\
			\cmidrule(lr){3-4} \cmidrule(lr){5-6} \cmidrule(lr){7-8}
			& & ATE & Sig. & PTTE & Sig. & Basic $-$ Ours & Basic $-$ Ground Truth \\
			& & & & & & (Observed) & (Expected) \\
			\midrule
			\textbf{Experiment A} & &\\
			\addlinespace[0.1em]
			& $M_1$ & \textcolor{ForestGreen}{Pos.} & Yes & \textcolor{BrickRed}{Neg.} & No & \textcolor{ForestGreen}{Pos.} & \textcolor{ForestGreen}{Pos.} \\
			& $M_2$ & \textcolor{ForestGreen}{Pos.} & No & \textcolor{BrickRed}{Neg.} & No & \textcolor{ForestGreen}{Pos.} & — \\
			& $M_3$ & \textcolor{BrickRed}{Neg.} & No & \textcolor{ForestGreen}{Pos.} & Yes & \textcolor{BrickRed}{Neg.} & —  \\
			\addlinespace[0.2em]
			\midrule
			\addlinespace[0.1em]
			\textbf{Experiment B} & & \\
			\addlinespace[0.1em]
			& $M_1$ & \textcolor{ForestGreen}{Pos.} & No & \textcolor{BrickRed}{Neg.} & No & \textcolor{ForestGreen}{Pos.} & \textcolor{ForestGreen}{Pos.} \\
			& $M_2$ & \textcolor{ForestGreen}{Pos.} & Yes & \textcolor{ForestGreen}{Pos.} & No & \textcolor{ForestGreen}{Pos.} & — \\
			& $M_3$ & \textcolor{ForestGreen}{Pos.} & No & \textcolor{BrickRed}{Neg.} & No & \textcolor{ForestGreen}{Pos.} & —  \\
			\bottomrule
		\end{tabular}
	\end{table}
	
	From the results, presented in Table \ref{tab:RealExperiments}, two observations follow. First, in both A and B our approach moves $M_1$ in the expected direction (the difference, Basic$-$Ours, is positive), recovering the correct \emph{direction of bias}. Second, for the decision metric, $M_3$, ignorance of spillover effects yields an alternate launch conclusion for Experiment A from what is concluded using our method. In Experiment B, both methods yield the same ``no effect'' conclusion on $M_3$ (non-significant), leaving the decision unchanged. Overall, it is reassuring that for $M_1$ our method consistently identifies the correct direction of spillover bias across both experiments.

	\section{Conclusion}\label{sec:conclusion}
	
	This paper studies eligibility-constrained bipartite experiments in which only a subset of treatment-side units can be randomized while interference propagates through all units. We formalize primary and secondary total treatment effects (PTTE, STTE), give identification under randomization within the eligible set, and propose flexible estimators that leverage generalized propensity scores and machine learning. Under a linear additive edges condition, we derive a projection linking treatment- and outcome-level estimands, enabling estimation on the typically smaller set of treatment units with aggregation to outcomes. In simulations and two case studies, accounting for interference yields effect estimates that differ materially from analyses that ignore spillovers, including a sign difference for one primary decision metric. These findings illustrate that, in settings with interaction across unit types, effect definitions and estimators that target the total impact at rollout can lead to different conclusions than conventional A/B analyses.
	
	\paragraph{Scope and limitations.}  
	Our framework is useful when (i) the bipartite structure is substantively meaningful, (ii) randomization is feasible within an eligible subset of treatment units, and (iii) one wishes to estimate the effect of full deployment to the eligible set. Several limitations should be highlighted.
	
	\begin{itemize}
		\item \emph{Network observability.} Identification and estimation rely on access to the interaction structure (e.g., weights $w_{ij}$ or close proxies). When the network is unavailable, only coarsely measured, or measured with substantial error, the proposed methods would not be applicable or may be biased. Designs and estimators for unknown networks exist \citep[e.g.,][]{yu2022estimating,cortez2022staggered,shirani2024causal,shirani2025can} but are not developed here.
		
		\item \emph{Exogenous network.} Assumption~\ref{ass:fixed-weights} requires that the network is unaffected by treatment. If assignment changes matching, congestion, or availability in ways that alter edges, the estimands here are not identified without additional modeling of joint assignment–network dynamics.
		
		\item \emph{Unconfoundedness and overlap.} With randomization inside the eligible set, Assumption~\ref{ass:weak-unconfoundedness} holds by design, but positivity (Assumption~\ref{ass:overlap}) can fail in sparse or highly skewed graphs. In such cases, models extrapolate to exposure regions with little support, increasing estimator sensitivity.
		
		\item \emph{Projection validity.} The projection between treatment and outcome levels requires linear additivity of edge-level outcomes (Assumption~\ref{ass:linear-additive}). When the target metric is non-additive (e.g., medians, capped rates, or composite satisfaction scores), the projection is invalid and estimation should remain at the metric’s native level.
		
		\item \emph{Finite-sample and algorithmic considerations.} The ensemble estimators depend on tuning choices and on the accuracy of generalized propensity features. Inference is implemented via bootstrap resampling; while practical, its coverage under complex dependence structures may deviate from nominal levels in small samples.
		
		\item \emph{External validity.} Empirical validation is limited to two experiments in a particular context. The sign and magnitude of spillovers are design- and environment-specific; conclusions need not transport to other systems, policies, or time periods.
	\end{itemize}
	
	\paragraph{Practical implications.}  
	When eligibility constraints coexist with cross-side interactions, analysts seeking the total effect of rollout should (i) collect and audit the interaction data used to form exposure, (ii) assess overlap in exposure distributions prior to modeling, (iii) report estimates at the natural measurement level of each metric and invoke projection only when Assumption~\ref{ass:linear-additive} is defensible, and (iv) accompany point estimates with diagnostic sensitivity analyses to plausible network misspecification. In settings where the network cannot be measured or Assumption~\ref{ass:fixed-weights} is doubtful, alternative designs that reduce interference (e.g., graph/cluster randomization) may be preferable \citep{ugander2013graph,eckles2017design}.
	
	\paragraph{Future directions.}  
	Promising extensions include: inference with partially observed or mismeasured networks (including formal sensitivity analyses); experimental designs that jointly learn the interaction structure while estimating PTTE/STTE; dynamic or sequential treatments where networks evolve; heterogeneity of total effects across network positions; and methods for non-additive outcomes. Finally, formal results on efficiency and inference under general dependence structures would help clarify the conditions under which the computational advantages of treatment-level estimation translate to statistical gains.
	
	Overall, the contribution is conceptual and methodological: we articulate estimands aligned with deployment in eligibility-constrained, bipartite settings and provide tools to estimate them under explicit assumptions. Equally important, we discuss the conditions under which these tools should be used with extra caution.

	\bibliographystyle{plainnat}
	\bibliography{references}
	
	\newpage
	
	\begin{appendix}
		
		\section{Key Notation}
		
		This section provides a summary of the key mathematical notation and variable definitions in Table~\ref{tab:notation}. 
		\begin{table}[H]
			\centering
			\small
			\caption{Key notation for eligibility-constrained bipartite experiments}
			\label{tab:notation}
			\begin{tabular}{cl}
				\toprule
				Symbol & Definition \\
				\midrule
				\multicolumn{2}{l}{\textbf{Units and Sets}} \\
				$\primaryset, \secondaryset$ & Primary (eligible) and secondary (ineligible) treatment units \\
				$\primaryoutcomeunits$ & Outcome units connected to $\ge 1$ primary unit \\
				$\mixeoutcomeunits$ & Outcome units connected to both primary \& secondary treatment units\\
				$\primetreatmentunits{i}$ & Set of primary units connected to outcome unit $i$ \\
				$\exposedcustomers{j}$ & Set of outcome units connected to treatment unit $j$ \\
				\midrule
				\multicolumn{2}{l}{\textbf{Treatment and Assignment}} \\
				$\bZ$ & Vector of treatment assignments (size is the total number of treatment-side units) \\
				$Z_j$ & Treatment assignment indicator for treatment-side unit $j$ \\
				$p$ & Treatment probability \\
				\midrule
				\multicolumn{2}{l}{\textbf{Network Structure}} \\
				$\bw$ & $n_{\outcome} \times n_{\treatment}$ matrix of connection weights \\
				$w_{ij}$ & Connection weight between outcome unit $i$ and treatment unit $j$ \\
				$\tilde{\bw}$ & Enriched graph structure including unit classifications\\
				\midrule
				\multicolumn{2}{l}{\textbf{Exposure Variables}} \\
				$E_i$ & Treatment exposure of outcome unit $i$ given assignment $\bZ$ \\
				$\Edirect_j, \Eindirect_j$ & Direct and indirect exposure of treatment unit $j$ \\
				$e_{i,\max}$ & Maximum exposure when all primary units are treated \\
				$n^{\primary}_i$ & Number of primary units connected to outcome unit $i$ \\
				$n^{\secondary}_i$ & Number of secondary units connected to outcome unit $i$ \\
				\midrule
				\multicolumn{2}{l}{\textbf{Outcome Variables}} \\
				$Y_i$ & Outcome for outcome unit $i$ \\
				$Y_j$ & Outcome for treatment unit $j$ \\
				$Y_{ij}$ & Edge-level outcome between units $i$ and $j$ \\
				$Y_{i,\secondary}$ & Outcomes attributable to secondary units \\
				$Y_{j,\secondary}$ & Secondary treatment unit outcome \\
				\midrule
				\multicolumn{2}{l}{\textbf{Functions and Parameters}} \\
				$\Phi, \Psi$ & Outcome functions for outcome and treatment units\\
				$r(E_i, \bw, X_i)$ & Generalized propensity score: $\prob(E_i | \bw, X_i)$ \\
				$X_i, X_j$ & Covariates for outcome and treatment units \\
				$\epsilon_i$ & Error term \\
				\midrule
				\multicolumn{2}{l}{\textbf{Estimands}} \\
				$\PTTE_{\outcome}$, $\STTE_{\outcome}$ & Primary and secondary total treatment effect at outcome-side \\
				$\PTTE_{\treatment}$, $\STTE_{\treatment}$ & Primary and secondary total treatment effect at treatment-side \\
				\midrule
				\multicolumn{2}{l}{\textbf{Mathematical Operators}} \\
				$\E[\cdot]$ & Expectation operator \\
				$\prob(\cdot)$ & Probability operator \\
				$\indicator(\cdot)$ & Indicator function (1 if true, 0 otherwise)\\
				\bottomrule
			\end{tabular}
		\end{table}
		
		\clearpage 
		
		\section{Additional Simulation Results}\label{sec:additiona-sims}
		
		Figures \ref{fig:ptteoutcome5spec}-\ref{fig:pttetreat5spec} show the boxplot of the ground truth and estimates of PTTE for outcome-level and treatment unit-level, respectively, for specifications 1-5. Similar quantities for STTE are presented in Figures \ref{fig:stteoutcome5spec}-\ref{fig:sttetreat5spec}.

		\begin{figure}[h]
			\centering
			\includegraphics[width=\linewidth]{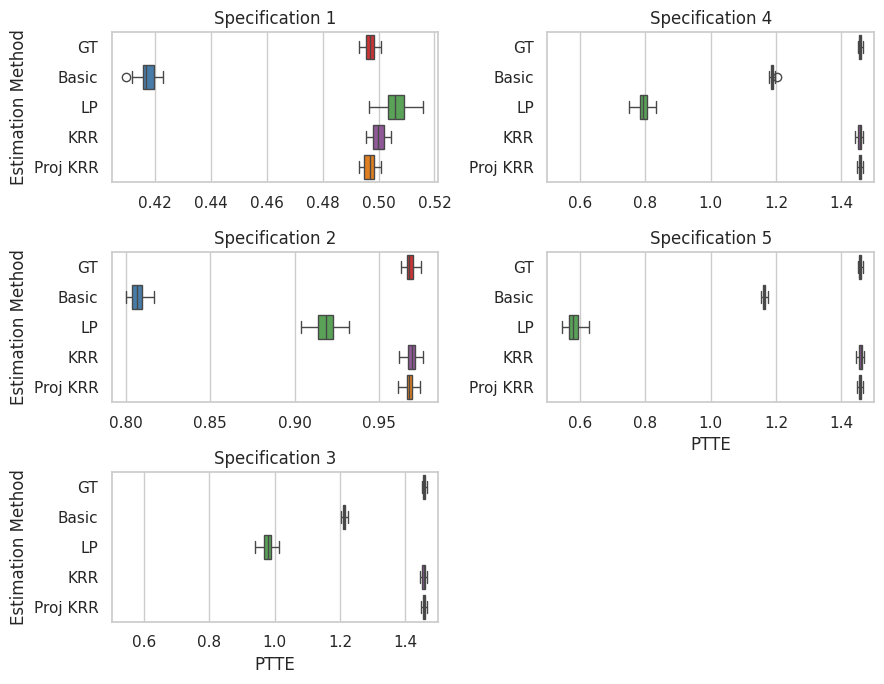}
			\caption{Box plot of PTTE at the outcome unit granularity}
			\label{fig:ptteoutcome5spec}
		\end{figure}
		\begin{figure}
			\centering    \includegraphics[width=\linewidth]{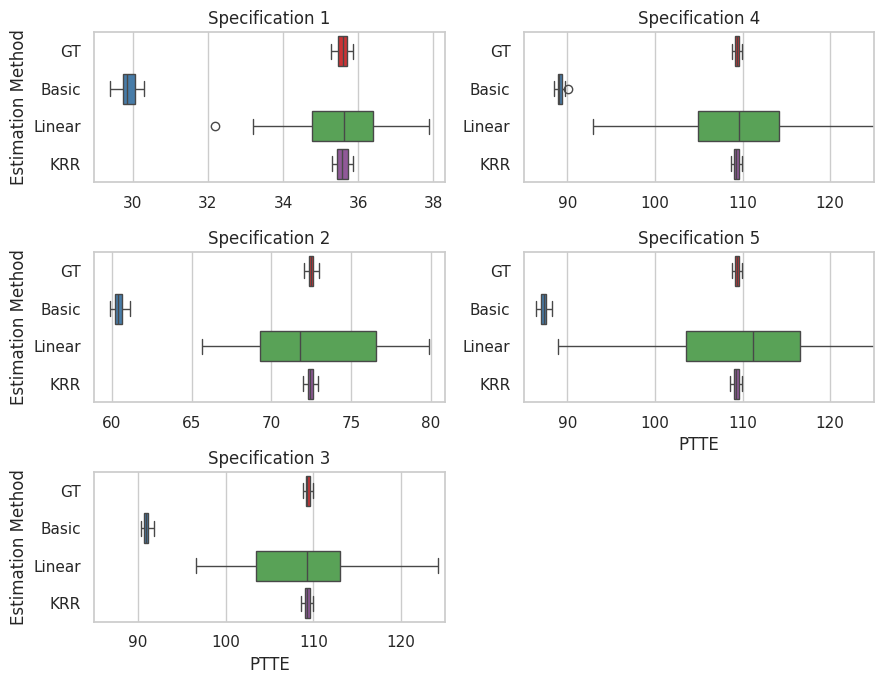}
			\caption{Box plot of PTTE at the treatment unit granularity}
			\label{fig:pttetreat5spec}
		\end{figure}
		\begin{figure}
			\centering \includegraphics[width=\linewidth]{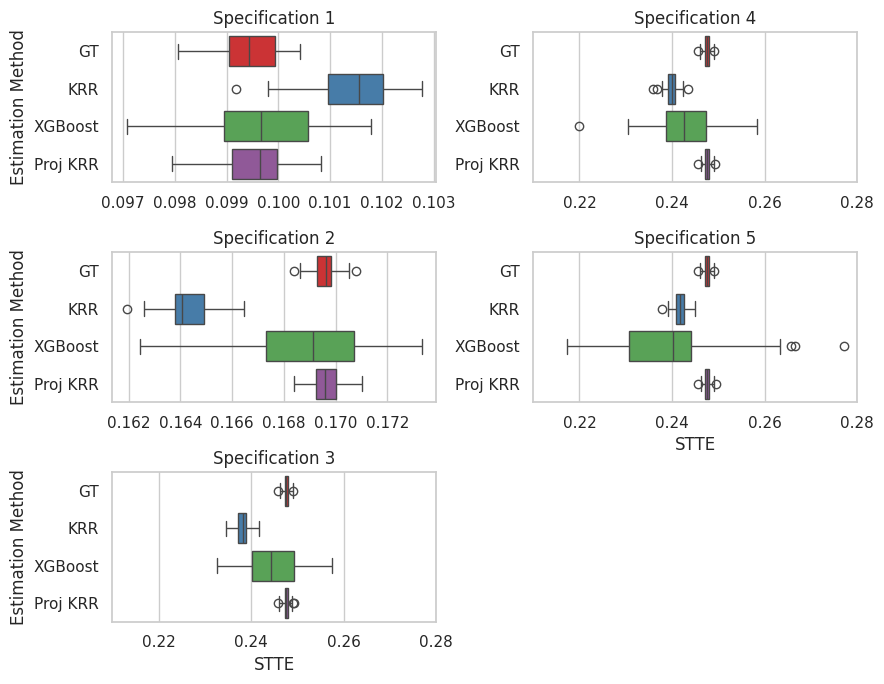}
			\caption{Box plot of STTE at the outcome unit  granularity}
			\label{fig:stteoutcome5spec}
		\end{figure}
		\begin{figure}
			\centering
			\includegraphics[width=\linewidth]{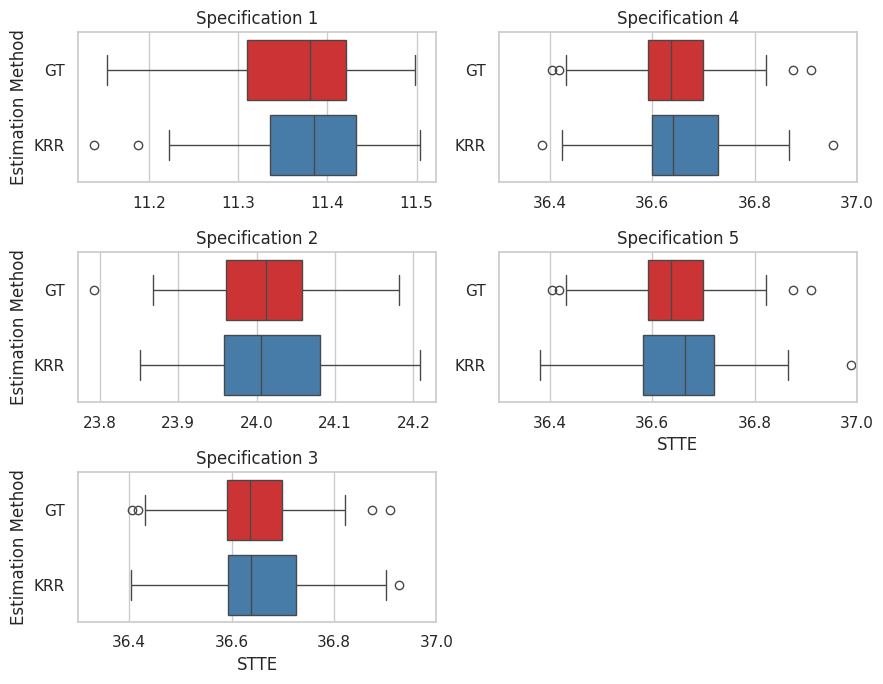}
			\caption{Box plot of STTE at the treatment unit granularity}
			\label{fig:sttetreat5spec}
		\end{figure}
		
	\end{appendix}
	
\end{document}